\documentclass[10pt,conference]{IEEEtran}
\IEEEoverridecommandlockouts
\usepackage{amsthm,amsmath,amssymb,mathtools,bm,etoolbox,float,hyperref}
\usepackage[dvipsnames]{xcolor}
\usepackage{graphicx,cite,cuted}
\usepackage{bigints}
\usepackage{caption}

\usepackage{stackengine}

\usepackage{mathrsfs,verbatim}
\usepackage{ellipsis}
\usepackage{tikz}
\usepackage{tikz-3dplot}
\usepackage[utf8]{inputenc}
\usepackage{pgfplots} 
\usepackage{pgfgantt}
\usepackage{pdflscape}
\pgfplotsset{compat=newest} 
\pgfplotsset{plot coordinates/math parser=false}
\pgfplotsset{every  tick/.style={black,},ylabel style={font=\tiny},xlabel style={font=\tiny},tick label style={font=\tiny},legend style= {font=\scriptsize},
minor x tick num=1,minor y tick num=1,xminorticks=true,yminorticks=true,}
\newlength\fheight
\newlength\fwidth
\usepackage{mathrsfs}

\usepackage{algorithm}
\usepackage{algpseudocode}
\makeatletter
\renewcommand{\Function}[2]{%
  \csname ALG@cmd@\ALG@L @Function\endcsname{#1}{#2}%
  \def\jayden@currentfunction{#1}%
}
\newcommand{\funclabel}[1]{%
  \@bsphack
  \protected@write\@auxout{}{%
    \string\newlabel{#1}{{\jayden@currentfunction}{\thepage}}%
  }%
  \@esphack
}

\usepackage{xinttools} 
\usepackage{tikz}
\usepackage{tkz-euclide}
\usetikzlibrary{calc}
\newtheorem{theorem}{Theorem}
\newtheorem{proposition}{Proposition}
\newtheorem{lemma}{Lemma}
\newtheorem{remark}{Remark}

\usepackage[dvipsnames]{xcolor}
\definecolor{cornellred}{rgb}{0.7, 0.11, 0.11}

\def\biglen{20cm} 
\tikzset{
  half plane/.style={ to path={
       ($(\tikztostart)!.5!(\tikztotarget)!#1!(\tikztotarget)!\biglen!90:(\tikztotarget)$)
    -- ($(\tikztostart)!.5!(\tikztotarget)!#1!(\tikztotarget)!\biglen!-90:(\tikztotarget)$)
    -- ([turn]0,2*\biglen) -- ([turn]0,2*\biglen) -- cycle}},
  half plane/.default={1pt}
}

\DeclareMathAlphabet{\pazocal}{OMS}{zplm}{m}{n}

\usepackage{caption}
\usepackage{multicol,subcaption}
\DeclareMathOperator*{\argmin}{arg\,min}

\begin{document}

\title{Low Complexity Hybrid Beamforming for mmWave Full-Duplex Integrated Access and Backhaul}



\author{Elyes Balti$^1$, Chris Dick$^2$ and Brian L. Evans$^1$ \\
$^1$6G@UT Research Center, Wireless Networking and Communications Group (WNCG) \\ The University of Texas at Austin, Austin, TX,
ebalti@utexas.edu, bevans@ece.utexas.edu \\
$^2$NVIDIA, Santa Clara, CA, 
cdick@nvidia.com \\
\thanks{E. Balti and B. L. Evans were supported by NVIDIA, an affiliate of the WNCG 6G@UT Research Center at UT Austin.}}

\maketitle

\begin{abstract}
We consider an integrated access and backhaul (IAB) node operating in full-duplex (FD) mode. We analyze simultaneous transmission from the New Radio gNB to the IAB node on the backhaul uplink, IAB node to a user equipment (UE) on the access downlink, and IAB transmitter to the IAB receiver on the self-interference (SI) channel. Our contributions include (1) a low complexity algorithm to jointly design the hybrid analog/digital beamformers for all three nodes to maximize the sum spectral efficiency of the access and backhaul links by canceling SI and maximizing received power; (2) derivation of all-digital beamforming and spectral efficiency upper bound for use in benchmarking; and (3) simulations to compare full vs. half duplex modes, hybrid vs. all-digital beamforming algorithms, proposed hybrid vs. conventional beamforming algorithms, and spectral efficiency upper bound. In simulations, the proposed algorithm shows significant reduction in SI power and increase in sum spectral efficiency.
\end{abstract}


\section{Introduction}
\IEEEPARstart{F}{uture} wireless networks are expected to have densely deployed basesetations (BSs) to support future applications, such as the Internet of Things, virtual/augmented reality, and vehicle-to-everything. However, traditional fiber backhauling is often unavailable or prohibitively expensive for carrier operators. Integrated access and backhaul (IAB) technology has emerged as a cost-effective alternative. In the case of IAB, only a few BSs are connected to the traditional wired infrastructures while the others relay the backhaul traffic wirelessly \cite{access,access1}. In a typical IAB framework, the access and backhaul links share the same frequency spectrum, which results in resource collision; thus, resource management is required to resolve this issue. Owing to the simplicity of implementation, many previous studies have incorporated half duplex (HD) constraints in their frameworks \cite{4}. In the HD IAB approach, the access and backhaul links must use the given radio resources orthogonally, be it in time or frequency. While this helps prevent collisions in the two links, it fails to exploit the full potential of the given radio resources. 

In contrast, a smarter IAB framework with full duplex
(FD) techniques may simply rule out the HD constraint. FD  systems have recently gained enormous attention in academia and industry due to its potential to reduce latency and double spectral efficiency in the link budget compared to the HD relays that transmit and receive in different time slots. These benefits make FD applicable in practice such as machine-to-machine and integrated access and backhaul which is currently proposed in 3GPP Release 17 \cite{backhaul1,release17,fdbackhaul}. 

Although FD brings many advantages, it suffers from loopback self-interference (SI), which is caused by the simultaneous transmission and reception over the same resource blocks. This loopback signal cannot be neglected as the  SI power can be several orders of magnitude stronger than the signal power received from the user equipment (UE), which can render FD systems dysfunctional \cite{zf}. To address this limitation, related work proposed robust beamforming design to suppress the SI signal and achieve acceptable spectral efficiency \cite{ian,adaptivebeamforming,backhaulenergy,fd1}. Authors in \cite{ian1} proposed a hybrid analog/digital beamforming for FD systems with limited dynamic range. In addition, authors in \cite{frequency} proposed a low complexity frequency-domain successive SI cancellation for FD radios. Authors in \cite{irs,irs1} proposed a robust beamforming design for an intelligent reflecting surfaces assisted FD multiuser systems to wipe out SI and improve sum spectral efficiency.

In this paper, we consider an FD IAB system. To address SI, we propose low complexity hybrid analog/digital beamforming to cancel SI, avoid analog-to-digital converter (ADC) saturation and maximize sum spectral efficiency of the access and backhaul links. We derive an all-digital solution and upper bound, and compare full vs. half duplex, hybrid vs. all-digital beamforming, conventional SI cancellation, and  upper bound.

Below, Section II describes the system model. Section III presents the optimization problem and beamforming design. Section IV gives numerical results. Section V concludes.

\textbf{Notation}: Bold lowercase $\mathbf{x}$ denotes column vectors, bold uppercase $\mathbf{X}$ denotes matrices, non-bold letters $x, X$ denote scalar values. Using this notation, $\| \mathbf{X}\|_F$ is the Frobenius norm, $\sigma_\ell(\mathbf{X})$ is the $\ell$-th singular value of $\mathbf{X}$ in decreasing order, $\mathrm{det}(\mathbf{X})$ denotes the determinant, $\mathrm{Tr}(\mathbf{X})$ denotes the trace, $\mathbf{X}^*$ is the Hermitian or conjugate transpose, $\mathbf{X}^{-1}$ denotes the inverse of a square non-singular matrix.
\section{System Model}
\begin{figure}[t]
    \centering
    \includegraphics[width=\linewidth]{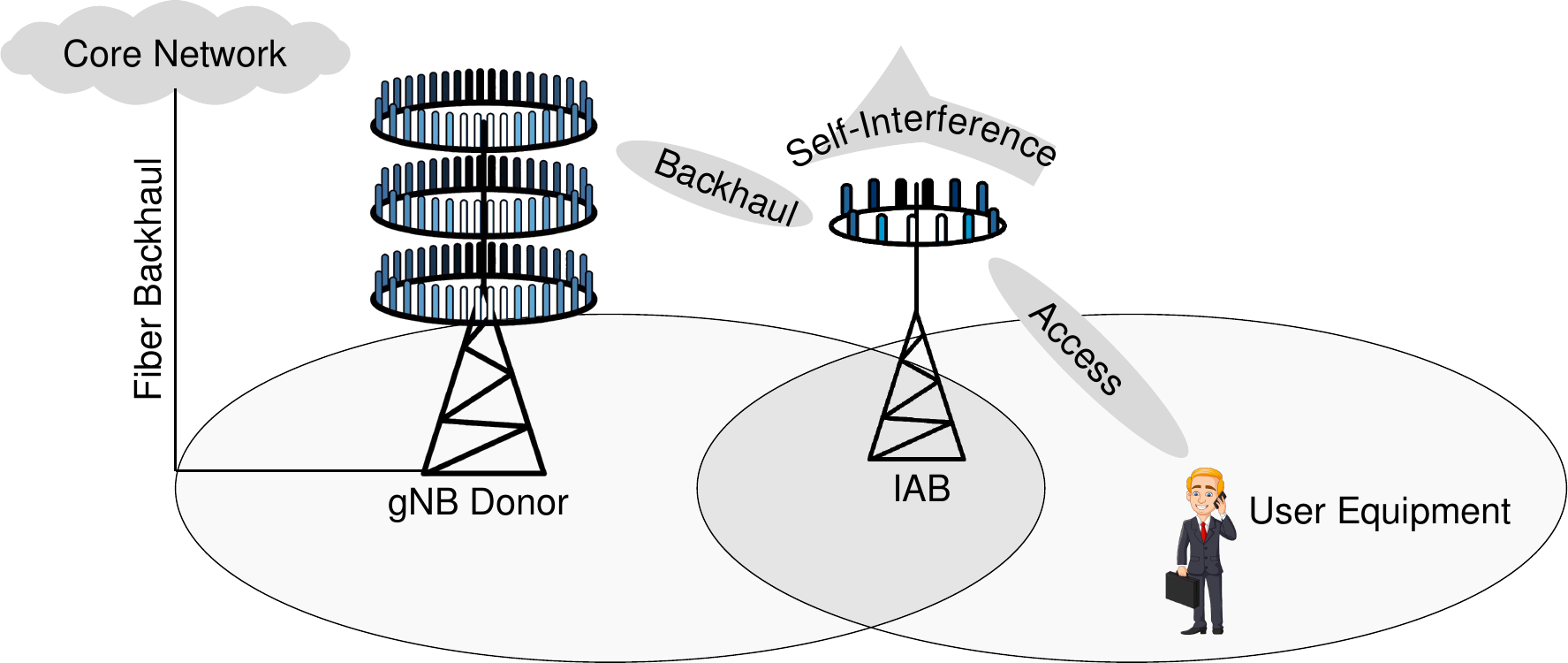}
    \caption{Full-duplex integrated access and backhaul (IAB) for a single-user case. The gNB donor, linked to the core network by fiber backhaul, communicates with the IAB node through wireless backhaul. The user equipment is served by the IAB node through the wireless access link. Simultaneous transmission and reception of the IAB node over the same time/frequency resources blocks incurs loopback self-interference. }
    \label{iab}
\end{figure}

Per Fig. \ref{iab}, the system transmits from gNB to IAB nodes on backhaul uplink, IAB node to user equipment on access downlink, and IAB transmitter to receiver on SI channel.

\subsection{Access and Backhaul Channel Models}

Per Fig. \ref{architecture}, the backhaul uplink channel, $\mathbf{H_b}$, and the downlink access channel, $\mathbf{H_a}$, each have the form
\begin{equation}\label{channel}
\begin{split}
\mathbf{H} = \sqrt{\frac{N_{\mathsf{RX}}N_{\mathsf{TX}}}{CR_c}} \sum_{c=0}^{C-1}\sum_{r_c=0}^{R_c-1} \alpha_{r_c} \mathbf{a}_{\mathsf{RX}}(\theta_{r_c})   \mathbf{a}_{\mathsf{TX}}^*(\phi_{r_c})    
\end{split}
\end{equation}
Where $C$ is number of clusters, $R_c$ is number of rays per cluster, and $\theta_{r_c}$ and $\phi_{r_c}$ are the angles of arrival (AoA) and departure (AoD) of the $r_c$-th ray, respectively. Each ray has a relative time delay $\tau_{r_c}$ and complex path gain $\alpha_{r_c}$. Also, $\mathbf{a}_{\mathsf{RX}}(\theta)$ and $\mathbf{a}_{\mathsf{TX}}(\phi)$ are the RX and TX antenna array response vectors, respectively. The array response vector is given by
\begin{equation}
\mathbf{a}_{\mathsf{X}}(\theta) = \frac{1}{\sqrt{N_{\mathsf{X}}}}\left[1,\mathsf{e}^{j\frac{2\pi d}{\lambda} \sin(\theta)},\ldots,\mathsf{e}^{j\frac{2\pi d}{\lambda}\left(N_{\mathsf{X}} -1\right)\sin(\theta)}  \right]^{T}.    
\end{equation}
Where $\mathsf{X}$ is the TX or RX and $N_{\mathsf{X}}$ is the number antennas.
\begin{figure*}[t]
    \centering
    \includegraphics[width=\linewidth]{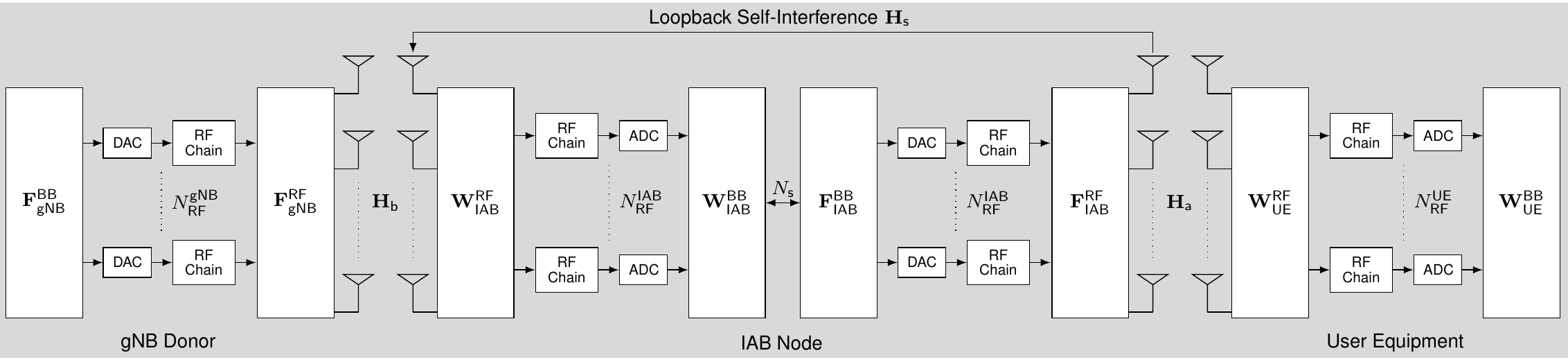}
    \caption{Basic abstraction of the hybrid analog/digital architecture of the full-duplex integrated access and backhaul system.
    The backhaul channel is between the gNB donor and IAB node, and
    the access channel is between the IAB node and the user equipment.}
    \label{architecture}
\end{figure*}

\subsection{Self-Interference Channel Model}
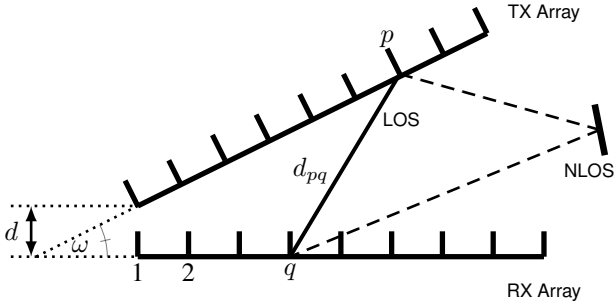
\begin{figure}[t]
\centering
\setlength\fheight{7.5cm}
\setlength\fwidth{7.5cm}

\usetikzlibrary{shapes.misc,shapes.geometric,shapes.symbols,positioning,shadings,automata}
\tikzstyle{XORgate} = [draw,circle]
\newcommand{\antena}{--++(3mm,0)--++(30:5mm)--++(-90:5mm)--++(150:5mm);}
\newcommand{\suma}{\Large$+$}

\definecolor{mycolor1}{rgb}{1.00000,0.00000,1.00000}%
\definecolor{orange}{rgb}{0.0,0.0,0.0}
\pgfdeclarelayer{background}
\pgfdeclarelayer{foreground}
\pgfsetlayers{background,main,foreground}
\usetikzlibrary{shapes,arrows}
\newcommand{\mx}[1]{\mathbf{\bm{#1}}} 
\newcommand{\vc}[1]{\mathbf{\bm{#1}}} 

\tikzstyle{sensor}=[draw, fill=yellow!30, text width=2em, 
    text centered, minimum height=2em, rounded corners]

\tikzstyle{chain}=[draw, fill=red!30, text width=2em, 
    text centered, minimum height=2em, rounded corners]
    
\tikzstyle{center1}=[draw=white,fill=white!20]   
    
\tikzstyle{ann} = [above, text width=5em]
\tikzstyle{node} = [sensor, text width=4em, fill=orange!50, 
    minimum height=19em, rounded corners]
   
\tikzstyle{circ} = [draw,circle,radius=0.5cm,fill=orange!50]
\def\blockdist{2.3}
\def\edgedist{2.5}
\def\antenna{%
    -- +(0mm,4.0mm) -- +(2.2mm,5.5mm) -- +(-2.2mm,5.5mm) -- +(0mm,4.0mm)
}

\usetikzlibrary{positioning}
\begin{tikzpicture}[thin,scale=0.675]


 \tkzDefPoint(0,0){A} \tkzDefPoint(1,0){B} \tkzDefPoint(1,.5){C} 
\tkzMarkAngle[fill= orange!40,size=1.4cm,opacity=.5](B,A,C)
\tkzLabelAngle[pos=0.9](B,A,C){$\omega$}  
\draw [-,dotted,line width=1pt] (0,0) to (2,0);                           
\draw [-,dotted,line width=1pt] (0,0) to (2,1);                           
\draw [-,orange,line width=2pt] (2,0) to (10,0); 
\draw [-,orange,line width=2pt] (2,0) to (2,0.5);
\draw [-,orange,line width=2pt] (3,0) to (3,0.5);
\draw [-,orange,line width=2pt] (4,0) to (4,0.5);
\draw [-,orange,line width=2pt] (5,0) to (5,0.5);
\draw [-,orange,line width=2pt] (6,0) to (6,0.5);
\draw [-,orange,line width=2pt] (7,0) to (7,0.5);
\draw [-,orange,line width=2pt] (8,0) to (8,0.5);
\draw [-,orange,line width=2pt] (9,0) to (9,0.5);
\draw [-,orange,line width=2pt] (10,0) to (10,0.5);
\node[align=right] at (10,-0.7){\scriptsize{\sffamily{RX Array}}};
\node[] at (2,-0.3){1};
\node[] at (3,-0.3){2};
\node[] at (5,-0.3){$q$};

\draw [-,orange,line width=2pt] (2,1) to (8.88,4.4); 
\draw [-,orange,line width=2pt] (2,1) to (1.75,1.5);
\draw [-,orange,line width=2pt] (2.86,1.4) to (2.61,1.9);
\draw [-,orange,line width=2pt] (3.72,1.9) to (3.47,2.4);
\draw [-,orange,line width=2pt] (4.58,2.3) to (4.33,2.8);
\draw [-,orange,line width=2pt] (5.44,2.7) to (5.19,3.2);
\draw [-,orange,line width=2pt] (6.3,3.1) to (6.05,3.6);
\draw [-,orange,line width=2pt] (7.16,3.6) to (6.91,4.1);
\node[] at (6.91,4.3){$p$};
\draw [-,orange,line width=2pt] (8.02,4) to (7.77,4.5);
\draw [-,orange,line width=2pt] (8.88,4.4) to (8.63,4.9);

\node[align=right] at (10,4.8){\scriptsize{\sffamily{TX Array}}};

\draw [-,line width=1.5pt] (5,0) to (7.16,3.6);
\node[] at (5.5,1.7){$d_{pq}~$~};
\node[] at (7.1,2.7){~\scriptsize{\sffamily{LOS}}};

\draw [-,black,line width=2.5pt] (11.2,2) to (11,3);
\node[] at (10.8,1.7){~\scriptsize{\sffamily{NLOS}}};

\draw [-,dash pattern={on 5pt off 3pt on 0pt off 0pt} ,line width=1pt] (7.16,3.6) to (11.1,2.5);
\draw [-,dash pattern={on 5pt off 3pt on 0pt off 0pt} ,line width=1pt] (11.1,2.5) to (5,0);


\draw [-,dotted,line width=1pt] (-0.5,0) to (0,0);                                                      
\draw [-,dotted,line width=1pt] (-0.5,1) to (2,1);                                                      
\draw [latex-latex,line width=1pt] (-0.1,0) to (-0.1,1);

\node[align=left] at (-0.4,0.5){$d~$};

\end{tikzpicture}
    \caption{ Relative position of TX and RX arrays at BS. Given that the TX and RX arrays are collocated, the far-field assumption that the signal impinges on the antenna
    array as a planar wave does not hold. Instead, for FD transceivers, it is more suitable to assume that the signal impinges on the array as a spherical wave for the near-field LOS channel.   }
     \label{array}
\end{figure}
Per Fig.~$\ref{array}$, the SI leakage at the BS is modeled by the channel matrix $\mathbf{H}_{\mathsf{s}}$. The separation, or transceiver gap, between  TX and RX arrays is defined by distance $d$ while the transceiver incline is determined by $\omega$.
The SI channel is decomposed into a static line-of-sight (LOS) channel modeled by $\mathbf{H}_{\mathsf{LOS}}$, which is derived from the geometry of the transceiver, and a non-line-of-sight (NLOS) channel described by $\mathbf{H}_{\mathsf{NLOS}}$ which follows the geometric channel model defined by (\ref{channel}). The ($q,p$)-th entry of the LOS SI leakage matrix can be written as 
\begin{equation}\label{eq2.2}
 [\mathbf{H}_{\mathsf{LOS}}]_{qp} = \frac{1}{d_{pq}}\mathsf{e}^{-j2\pi\frac{d_{pq}}{\lambda}}    
\end{equation}
Where $d_{pq}$ is the distance between the $p$-th antenna in the TX array and $q$-th antenna in the RX array at BS given by (\ref{distance}). The aggregate SI channel matrix can be obtained by
\begin{equation}\label{eq2.3}
\mathbf{H}_{\mathsf{s}} = \underbrace{\sqrt{\frac{\kappa}{\kappa+1}}\mathbf{H}_{\mathsf{LOS}}}_{\textsf{Near-Field}} + \underbrace{\sqrt{\frac{1}{\kappa+1}}\mathbf{H}_{\mathsf{NLOS}}}_{\textsf{Far-Field}}   
\end{equation}
Where $\kappa$ is the Rician factor.
\begin{figure*}[t]
\begin{equation}\label{distance}
\begin{split}
d_{pq} = \sqrt{ \left( \frac{d}{\tan(\omega)} + (q-1)\frac{\lambda}{2} \right)^2 + \left( \frac{d}{\sin(\omega)}+(p-1)\frac{\lambda}{2} \right)^2 -2\left( \frac{d}{\tan(\omega)} + (q-1)\frac{\lambda}{2} \right)\left( \frac{d}{\sin(\omega)}+(p-1)\frac{\lambda}{2} \right)\cos(\omega)   }    
\end{split}
\end{equation}
\vspace*{-.5cm}
\end{figure*}
\subsection{Signal Model}
Received signals at the IAB ($\mathbf{y}_{\mathsf{b}}$) and UE ($\mathbf{y}_{\mathsf{a}}$) are given by
\begin{equation}
\begin{split}
\mathbf{y}_{\mathsf{b}} =&  \underbrace{\sqrt{\rho_{\mathsf{b}}}\mathbf{W}_{\mathsf{IAB}}^*\mathbf{H}_{\mathsf{b}}\mathbf{F}_{\mathsf{gNB}}\mathbf{s}_{\mathsf{b}}}_{\textsf{Desired Signal}} + \underbrace{\sqrt{\rho_{\mathsf{s}}} \mathbf{W}_{\mathsf{IAB}}^*\mathbf{H}_{\mathsf{s}}\mathbf{F}_{\mathsf{IAB}}\mathbf{s}_{\mathsf{a}}}_{\textsf{Self-Interference Signal}}\\& +  \underbrace{\mathbf{W}_{\mathsf{IAB}}^* \mathbf{n}_{\mathsf{IAB}}}_{\textsf{AWGN}}
\end{split}
\end{equation}
\begin{equation}
\begin{split}
\mathbf{y}_{\mathsf{a}} =&  \underbrace{\sqrt{\rho_{\mathsf{a}}}\mathbf{W}_{\mathsf{UE}}^*\mathbf{H}_{\mathsf{a}}\mathbf{F}_{\mathsf{IAB}}\mathbf{s}_{\mathsf{a}}}_{\textsf{Desired Signal}} +  \underbrace{\mathbf{W}_{\mathsf{UE}}^* \mathbf{n}_{\mathsf{UE}}}_{\textsf{AWGN}}
\end{split}
\end{equation}
Where $\mathbf{W}_{\mathsf{IAB}} \in \mathbb{C}^{N_{\mathsf{IAB}}\times N_{\mathsf{s}}}$ and $\mathbf{F}_{\mathsf{IAB}} \in \mathbb{C}^{N_{\mathsf{IAB}}\times N_{\mathsf{s}}}$ are the all-digital combiner and precoder at the IAB node, respectively.  $\mathbf{W}_{\mathsf{UE}} \in \mathbb{C}^{N_{\mathsf{UE}}\times N_{\mathsf{s}}}$ and $\mathbf{F}_{\mathsf{gNB}} \in \mathbb{C}^{N_{\mathsf{gNB}}\times N_{\mathsf{s}}}$ being the all-digital combiner and precoder at the UE and gNB, respectively. Also, $N_{\mathsf{s}}$ is the number of spatial streams and $N_{\mathsf{X}}$ is the number of antennas at node $\mathsf{X}$.

\section{Beamforming Design}
The objective of designing of the beamformers is to maximize the received power for backhaul and access links and simultaneously reject the SI. In this work, we propose a hybrid analog/digital beamforming design wherein large amount of SI is suppressed in the analog domain to avoid the ADC saturation while residual SI is wiped out in the digital domain.
\subsection{Hybrid Beamforming: Analog Stage}
In this stage, we proceed to design the analog combiner $\mathbf{W}_{\mathsf{IAB}}^{\mathsf{RF}} \in \mathbb{C}^{N_{\mathsf{IAB}}\times N_{\mathsf{RF}}^{\mathsf{IAB}}}$ and precoder $\mathbf{F}_{\mathsf{IAB}}^{\mathsf{RF}} \in \mathbb{C}^{N_{\mathsf{IAB}}\times N_{\mathsf{RF}}^{\mathsf{IAB}}}$ at the IAB node as well as the analog combiner at UE $\mathbf{W}_{\mathsf{UE}}^{\mathsf{RF}} \in \mathbb{C}^{N_{\mathsf{UE}}\times N_{\mathsf{RF}}^{\mathsf{UE}}}$ and the analog precoder at the gNB $\mathbf{F}_{\mathsf{gNB}}^{\mathsf{RF}} \in \mathbb{C}^{N_{\mathsf{gNB}}\times N_{\mathsf{RF}}^{\mathsf{gNB}}}$, where $N_{\mathsf{RF}}^{\mathsf{X}}$ is the number of RF chains at node $\mathsf{X}$. To avoid the ADC saturation, large amount of SI has to be rejected in the analog domain which consequently requires a robust design. The covariance matrix of the precoded SI and noise  at the IAB node is expressed by
\begin{equation}\label{Rmatrix}
\mathbf{R}_{\mathsf{IAB}} = \rho_{\mathsf{s}}\mathbf{H}_{\mathsf{s}}\mathbf{F}_{\mathsf{IAB}}^{\mathsf{RF}} \mathbf{F}_{\mathsf{IAB}}^{\mathsf{RF}*} \mathbf{H}_{\mathsf{s}}^* + \sigma^2\mathbf{I}_{N_{\mathsf{IAB}}} 
\end{equation}
Where $\sigma^2$ is the noise variance. Our objective is to jointly design the analog combiners $\mathbf{W}_{\mathsf{IAB}}^{\mathsf{RF}}$, $\mathbf{W}_{\mathsf{UE}}^{\mathsf{RF}}$ and precoders $\mathbf{F}_{\mathsf{gNB}}^{\mathsf{RF}}$, $\mathbf{F}_{\mathsf{IAB}}^{\mathsf{RF}}$  to minimize the SI power at the IAB node and preserve the dimension of the signal space, i.e., $\mathsf{rank}\left(\mathbf{W}_{\mathsf{IAB}}^{\mathsf{RF}*}\mathbf{H}_{\mathsf{b}} \mathbf{F}_{\mathsf{gNB}}^{\mathsf{RF}} \right) = \mathrm{min}\left(N_{\mathsf{RF}}^{\mathsf{IAB}},N_{\mathsf{RF}}^{\mathsf{gNB}}\right)$ and $\mathsf{rank}\left(\mathbf{W}_{\mathsf{UE}}^{\mathsf{RF}*}\mathbf{H}_{\mathsf{a}} \mathbf{F}_{\mathsf{IAB}}^{\mathsf{RF}} \right) = \mathrm{min}\left(N_{\mathsf{RF}}^{\mathsf{UE}},N_{\mathsf{RF}}^{\mathsf{IAB}}\right)$. We formulate the optimization problem accordingly
\begin{equation}\label{problem1}
\begin{split}
\mathscr{P}_1:& \min\limits_{\mathbf{W}_{\mathsf{IAB}}^{\mathsf{RF}}} \mathrm{Tr}\left( \mathbf{W}_{\mathsf{IAB}}^{\mathsf{RF}*}\mathbf{R}_{\mathsf{IAB}}\mathbf{W}_{\mathsf{IAB}}^{\mathsf{RF}} \right)\\
\mathsf{s.t.}~& \mathbf{W}_{\mathsf{IAB}}^{\mathsf{RF}*}\mathbf{H}_{\mathsf{b}}\mathbf{F}_{\mathsf{gNB}}^{\mathsf{RF}} = \alpha \mathbf{I}_{N_{\mathsf{RF}}^{\mathsf{IAB}}}
\end{split}
\end{equation}
Where $\mathbf{R}_{\mathsf{IAB}}$ is a positive definite matrix ($\mathbf{R}_{\mathsf{IAB}} > 0$) and $\alpha = 1/\sqrt{\mathrm{Tr}\left( \mathbf{W}_{\mathsf{IAB}}^{\mathsf{RF}*} \mathbf{W}_{\mathsf{IAB}}^{\mathsf{RF}}\right)}$ is a power normalization coefficient.

To design the analog precoder at the IAB $\mathbf{F}_{\mathsf{IAB}}^{\mathsf{RF}}$, we proceed similarly as $\mathbf{W}_{\mathsf{IAB}}^{\mathsf{RF}}$. The covariance matrix of the combined SI and noise is expressed by
\begin{equation}\label{Smatrix}
\mathbf{S}_{\mathsf{IAB}} = \rho_{\mathsf{s}}  \mathbf{H}_{\mathsf{s}}^* \mathbf{W}_{\mathsf{IAB}}^{\mathsf{RF}} \mathbf{W}_{\mathsf{IAB}}^{\mathsf{RF*}} \mathbf{H}_{\mathsf{s}} + \sigma^2 \mathbf{I}_{N_\mathsf{IAB}}  
\end{equation}
Where ($\mathbf{S}_{\mathsf{IAB}} > 0$) is a positive definite matrix. Then, we formulate the problem accordingly
\begin{equation}\label{problem2}
\begin{split}
\mathscr{P}_2:& \min\limits_{\mathbf{F}_{\mathsf{IAB}}^{\mathsf{RF}}} \mathrm{Tr}\left( \mathbf{F}_{\mathsf{IAB}}^{\mathsf{RF}*}\mathbf{S}_{\mathsf{IAB}}\mathbf{F}_{\mathsf{IAB}}^{\mathsf{RF}} \right)\\
\mathsf{s.t.}~& \mathbf{W}_{\mathsf{UE}}^{\mathsf{RF}*}\mathbf{H}_{\mathsf{a}}\mathbf{F}_{\mathsf{IAB}}^{\mathsf{RF}} = \beta \mathbf{I}_{N_{\mathsf{RF}}^{\mathsf{IAB}}}
\end{split}
\end{equation}
Where $\beta = 1/\sqrt{\mathrm{Tr}\left( \mathbf{F}_{\mathsf{IAB}}^{\mathsf{RF}*} \mathbf{F}_{\mathsf{IAB}}^{\mathsf{RF}}\right)}$ is a power normalization coefficient.
\begin{theorem}\label{theorem1}
The optimal analog combiner and precoder at the IAB node, solutions to the problems (\ref{problem1}) and (\ref{problem2}) are expressed by
\begin{equation}\label{rfcombiab}
\mathbf{W}_{\mathsf{IAB}}^{\mathsf{RF}} = \alpha \mathbf{R}_{\mathsf{IAB}}^{-1}\mathbf{H}_{\mathsf{b}}\mathbf{F}_{\mathsf{gNB}}^{\mathsf{RF}}\left( \mathbf{F}_{\mathsf{gNB}}^{\mathsf{RF*}}\mathbf{H}_{\mathsf{b}}^* \mathbf{R}_{\mathsf{IAB}}^{-1} \mathbf{H}_{\mathsf{b}} \mathbf{F}_{\mathsf{gNB}}^{\mathsf{RF}}   \right)^{-1}    
\end{equation}
\begin{equation}\label{rfpreciab}
\mathbf{F}_{\mathsf{IAB}}^{\mathsf{RF}} = \beta \mathbf{S}_{\mathsf{IAB}}^{-1}\mathbf{H}_{\mathsf{a}}^*\mathbf{W}_{\mathsf{UE}}^{\mathsf{RF}*}\left( \mathbf{W}_{\mathsf{UE}}^{\mathsf{RF*}}\mathbf{H}_{\mathsf{a}} \mathbf{S}_{\mathsf{IAB}}^{-1} \mathbf{H}_{\mathsf{a}}^* \mathbf{W}_{\mathsf{UE}}^{\mathsf{RF}}   \right)^{-1}    
\end{equation}
\end{theorem}
\begin{proof}
The proof of Theorem \ref{theorem1} is provided in Appendix \ref{appendixtheorem1}.
\end{proof}
In this design, the analog precoder at gNB $\mathbf{F}_{\mathsf{gNB}}^{\mathsf{RF}}$ and combiner at UE $\mathbf{F}_{\mathsf{gNB}}^{\mathsf{RF}}$ can be selected regardless of Problems (\ref{problem1}) and (\ref{problem2}). 
\begin{proposition}
The analog combiner at the UE that minimizes the SI and hence the Mean Square Error (MSE) is the Wiener filter or Linear Minimum (LMMSE) receiver $\mathbf{W}_{\mathsf{MMSE}}$. The filter design problem can be defined as
\begin{equation}
\begin{split}
\mathscr{P}_3:&\mathbf{W}_{\mathsf{MMSE}} = \argmin\limits_{\mathbf{W}}\mathbb{E}\left[\|\mathbf{s}-\mathbf{y}\|_2^2 \right]
\end{split}
\end{equation}
For the analog precoder at the gNB, we adopt the Regularized Zero-Forcing filter $\mathbf{F}_{\mathsf{RegZF}}$. The expressions of the analog combiner and precoder at the UE and gNB, respectively, are given by 
\begin{equation}\label{rfcombue}
\mathbf{W}_{\mathsf{UE}}^{\mathsf{RF}} = \left( \mathbf{H}_{\mathsf{a}}\mathbf{F}_{\mathsf{IAB}}^{\mathsf{RF}}\mathbf{F}_{\mathsf{IAB}}^{\mathsf{RF}*} \mathbf{H}_{\mathsf{a}}^* + \frac{N_{\mathsf{UE}}}{\mathsf{SNR}_{\mathsf{a}}}   \mathbf{I}_{N_\mathsf{UE}} \right)^{-1} \mathbf{H}_{\mathsf{a}}\mathbf{F}_{\mathsf{IAB}}^{\mathsf{RF}}     
\end{equation}
\begin{equation}\label{rfprecgnb}
\mathbf{F}_{\mathsf{gNB}}^{\mathsf{RF}} = \left( \mathbf{H}_{\mathsf{b}}^*\mathbf{W}_{\mathsf{IAB}}^{\mathsf{RF}}\mathbf{W}_{\mathsf{IAB}}^{\mathsf{RF}*} \mathbf{H}_{\mathsf{b}} + \frac{N_{\mathsf{IAB}}}{\mathsf{SNR}_{\mathsf{b}}}   \mathbf{I}_{N_\mathsf{IAB}} \right)^{-1} \mathbf{H}_{\mathsf{b}}^*\mathbf{W}_{\mathsf{IAB}}^{\mathsf{RF}}     
\end{equation}
\end{proposition}
Where $\mathsf{SNR}_{\mathsf{x}} = \frac{\rho_\mathsf{x}}{\sigma^2}$, $\mathsf{x} \in \{a, b\}$.

The analog beamformers designed in Eqs.~(\ref{rfcombiab}-\ref{rfprecgnb}) are unconstrained solutions, i.e., they do not satisfy the constant amplitude (CA) constraint. To satisfy such constraint, they have to be projected onto the subspace of the CA constraint. Equivalently, the unconstrained solutions are updated as follows
\begin{equation}\label{caconstraint}
\mathbf{X}_{\mathsf{RF}} \longleftarrow \frac{1}{\sqrt{N}}\exp\left(\mathsf{i} \angle{\mathbf{X}_{\mathsf{RF}}} \right)    
\end{equation}
Where $N$ and $\angle\mathbf{X}$ are the number of rows and angles of the complex matrix $\mathbf{X}$, respectively. 

\subsection{Hybrid Beamforming: Digital Stage}
Once the analog beamformers are designed to reject large amount of SI to avoid the ADC saturation, the analog cancellation is not perfect, i.e., there are some residual SI left over after the analog stage. The digital beamformers which are interpreted as the last line of defense come to further remove this residual SI. 
\begin{theorem}\label{theorem2}
The optimal digital beamformer $\mathbf{X}_{\mathsf{BB}}$ can be expressed in terms of the analog beamformer $\mathbf{X}_{\mathsf{RF}}$ as follows. We first apply the SVD $\mathbf{X}_{\mathsf{RF}} = \mathbf{U}_{\mathsf{RF}} \mathbf{S}_{\mathsf{RF}} \mathbf{V}_{\mathsf{RF}}^*$. Second we express $\mathbf{X}_{\mathsf{BB}} = \mathbf{V}_{\mathsf{RF}}\mathbf{S}_{\mathsf{RF}}^{-1}\mathbf{Q}_{\star}$, where the columns of $\mathbf{Q}_{\star} \in \mathbb{C}^{M \times N}$ comprise the $N$ dominant left singular vectors of $\mathbf{U}_{\mathsf{RF}}^* \mathbf{A}$. Note that $\mathbf{X}_{\mathsf{RF}} \in \mathbb{C}^{M\times L}$, $\mathbf{X}_{\mathsf{BB}} \in \mathbb{C}^{L\times N}$ and $\mathbf{A} \in \mathbb{C}^{M\times N}$. 
\end{theorem}
\begin{proof}
The proof of Theorem \ref{theorem2} is reported in Appendix \ref{appendixtheorem2}.
\end{proof}

\subsection{All-Digital Beamforming}
In mmWave communications, all-digital beamforming using one RF chain per antenna with full precision data converters is not a practical design. Although it achieves high spectral efficiency, it is not energy efficient. However, such a design may serve as a benchmarking tool to measure the efficacy of the proposed hybrid beamforming design.  To this end, we design an all-digital beamformer to cancel SI and maximize the sum spectral efficiency by extending the routine we are using for analog beamforming design. The first extension is the all-digital beamformer $\mathbf{X} \in \mathbb{C}^{M\times N}$ would have different dimensions compared to the analog beamfomers, where $M$ and $N$ are the numbers of antennas and spatial streams, respectively. The second extension is that the all-digital beamformer design is unconstrained; i.e., the CA constraint does not exist for such a design.
\begin{remark}
Given the all-digital beamformer solution $\mathbf{X} \in \mathbb{C}^{M\times N}$, $M$ and $N$ are the number of antennas and spatial streams, respectively. $M$ should be large enough to sustain $N$ spatial streams and the remaining $P = M-N$ degrees of freedom should dedicated to suppress the SI.
\end{remark}

We introduce the expressions of the spectral efficiency for the backhaul and access links, respectively, as follows
\begin{equation}\label{ratebackhaul}
\mathcal{I}_{\mathsf{b}} = \log\det\left(\mathbf{I}_{N_\mathsf{s}} +\rho_{\mathsf{b}}\mathbf{W}_{\mathsf{IAB}}^*\mathbf{H}_\mathsf{b}\mathbf{F}_{\mathsf{gNB}}\mathbf{Q}_\mathsf{b}^{-1} \mathbf{F}_{\mathsf{gNB}}^*\mathbf{H}_\mathsf{b}^*  \mathbf{W}_{\mathsf{IAB}}\right)    
\end{equation}
\begin{equation}\label{rateaccess}
\mathcal{I}_{\mathsf{a}} = \log\det\left(\mathbf{I}_{N_\mathsf{s}} +\rho_{\mathsf{a}}\mathbf{W}_{\mathsf{UE}}^*\mathbf{H}_\mathsf{a}\mathbf{F}_{\mathsf{IAB}}\mathbf{Q}_\mathsf{a}^{-1} \mathbf{F}_{\mathsf{IAB}}^*\mathbf{H}_\mathsf{a}^*  \mathbf{W}_{\mathsf{UE}}\right)    
\end{equation}
Where $\mathbf{Q}_\mathsf{b}$ is the covariance matrix of the SI and noise power for the backhaul link and $\mathbf{Q}_\mathsf{a}$ is the covariance matrix of the noise power for the access link, respectively given by 
\begin{equation}
 \mathbf{Q}_\mathsf{b} = \rho_\mathsf{s} \mathbf{W}_{\mathsf{IAB}}^* \mathbf{H}_{\mathsf{s}} \mathbf{F}_{\mathsf{IAB}} \mathbf{F}_{\mathsf{IAB}}^*\mathbf{H}_{\mathsf{s}}^*\mathbf{W}_{\mathsf{IAB}}+ \sigma^2 \mathbf{W}_{\mathsf{IAB}}^*\mathbf{W}_{\mathsf{IAB}}
\end{equation}
\begin{equation}
 \mathbf{Q}_\mathsf{a} = \sigma^2 \mathbf{W}_{\mathsf{IAB}}^*\mathbf{W}_{\mathsf{IAB}}
\end{equation}

\begin{lemma}
For the interference-free case, the optimal beamformers diagonalize the channel. By applying the SVD on the channel, we retrieve the singular values and extract the first $N_{\mathsf{s}}$ modes associated with the spatial streams. The upper bound for backhaul or access link is given by
\begin{equation}\label{upperbound}
\mathcal{I}_{\mathsf{Bound}}=  \sum_{\ell=0}^{N_{\mathsf{s}}-1}\log\left(1 +  \sigma_{\ell}\left(\mathbf{H}  \right)^2 \mathsf{SNR} \right)  
\end{equation}
\end{lemma}

\begin{algorithm}[t]
\caption{Hybrid Beamforming Design}
\label{hybrid-beamforming}
\begin{algorithmic}[1]
\Function{Digital}{$\mathbf{X}_{\mathsf{RF}},\mathbf{A},N$} \funclabel{alg:a2}
\label{alg:a-line2}
    \State Compute SVD $\mathbf{X}_{\mathsf{RF}} = \mathbf{U}_{\mathsf{RF}} \mathbf{S}_{\mathsf{RF}}\mathbf{V}^*_{\mathsf{RF}}$
        \
    \State $\mathbf{Q} \gets N~\text{Dominant left singular vectors of } \mathbf{U}^*_{\mathsf{RF}}\mathbf{A}$
        
    \State $\mathbf{X}_{\mathsf{BB}} \gets \mathbf{V}_{\mathsf{RF}} \mathbf{S}^{-1}_{\mathsf{RF}} \mathbf{Q} $ 
    \State \Return $\mathbf{X}_{\mathsf{BB}}$
\EndFunction
\Statex
\State \textbf{Input} $\mathbf{H}_{\mathsf{s}} ,\mathbf{H}_{\mathsf{a}} ,\mathbf{H}_{\mathsf{b}}$
\State \textbf{Initialize} $\mathbf{F}_{\mathsf{gNB}}^{\mathsf{RF}}$, $\mathbf{F}_{\mathsf{IAB}}^{\mathsf{RF}}$, $\mathbf{W}_{\mathsf{IAB}}^{\mathsf{RF}}$, $\mathbf{W}_{\mathsf{UE}}^{\mathsf{RF}}$, $\mathbf{F}_{\mathsf{gNB}}^{\mathsf{BB}} $, $\mathbf{F}_{\mathsf{IAB}}^{\mathsf{BB}} $, $\mathbf{W}_{\mathsf{IAB}}^{\mathsf{BB}} $, $\mathbf{W}_{\mathsf{UE}}^{\mathsf{BB}} $
\State Set $\mathbf{W}_{\mathsf{IAB}} \gets \mathbf{W}_{\mathsf{IAB}}^{\mathsf{RF}}\mathbf{W}_{\mathsf{IAB}}^{\mathsf{BB}}$, $\mathbf{F}_{\mathsf{IAB}} \gets \mathbf{F}_{\mathsf{IAB}}^{\mathsf{RF}}\mathbf{F}_{\mathsf{IAB}}^{\mathsf{BB}}$, $\mathbf{W}_{\mathsf{IAB}} \gets \mathbf{W}_{\mathsf{UE}}^{\mathsf{RF}}\mathbf{W}_{\mathsf{UE}}^{\mathsf{BB}}$ and $\mathbf{F}_{\mathsf{gNB}} \gets \mathbf{F}_{\mathsf{gNB}}^{\mathsf{RF}}\mathbf{F}_{\mathsf{gNB}}^{\mathsf{BB}}$. 
\State Compute $\mathbf{R}_{\mathsf{IAB}}$ and $\mathbf{S}_{\mathsf{IAB}}$ from (\ref{Rmatrix}) and (\ref{Smatrix}).
\State Obtain $\mathbf{W}_{\mathsf{IAB}}^{\mathsf{RF}}$ and $\mathbf{F}_{\mathsf{IAB}}^{\mathsf{RF}}$ from (\ref{rfcombiab}) and (\ref{rfpreciab}).
\State Obtain $\mathbf{W}_{\mathsf{UE}}^{\mathsf{RF}}$ and $\mathbf{F}_{\mathsf{gNB}}^{\mathsf{RF}}$ from (\ref{rfcombue}) and (\ref{rfprecgnb}).
\State Project the analog beamformers on the CA subspace (\ref{caconstraint}).
\State $\mathbf{W}_{\mathsf{IAB}}^{\mathsf{BB}} \gets $ \textsc{Digital}$\left(\mathbf{W}_{\mathsf{IAB}}^{\mathsf{RF}},\mathbf{H}_{\mathsf{b}}\mathbf{F}_{\mathsf{gNB}},N_{\mathsf{s}}\right)$.
\State $\mathbf{F}_{\mathsf{IAB}}^{\mathsf{BB}} \gets $ \textsc{Digital}$\left(\mathbf{F}_{\mathsf{IAB}}^{\mathsf{RF}},\mathbf{H}_{\mathsf{a}}^*\mathbf{W}_{\mathsf{UE}},N_{\mathsf{s}}\right)$.
\State $\mathbf{W}_{\mathsf{UE}}^{\mathsf{BB}} \gets $ \textsc{Digital}$\left(\mathbf{W}_{\mathsf{UE}}^{\mathsf{RF}},\mathbf{H}_{\mathsf{a}}\mathbf{F}_{\mathsf{IAB}},N_{\mathsf{s}}\right)$.
\State $\mathbf{F}_{\mathsf{gNB}}^{\mathsf{BB}} \gets $ \textsc{Digital}$\left(\mathbf{F}_{\mathsf{gNB}}^{\mathsf{RF}},\mathbf{H}_{\mathsf{b}}^*\mathbf{W}_{\mathsf{IAB}},N_{\mathsf{s}}\right)$.
\State Set $\mathbf{W}_{\mathsf{IAB}} \gets \mathbf{W}_{\mathsf{IAB}}^{\mathsf{RF}}\mathbf{W}_{\mathsf{IAB}}^{\mathsf{BB}}$, $\mathbf{F}_{\mathsf{IAB}} \gets \mathbf{F}_{\mathsf{IAB}}^{\mathsf{RF}}\mathbf{F}_{\mathsf{IAB}}^{\mathsf{BB}}$, $\mathbf{W}_{\mathsf{IAB}} \gets \mathbf{W}_{\mathsf{UE}}^{\mathsf{RF}}\mathbf{W}_{\mathsf{UE}}^{\mathsf{BB}}$ and $\mathbf{F}_{\mathsf{gNB}} \gets \mathbf{F}_{\mathsf{gNB}}^{\mathsf{RF}}\mathbf{F}_{\mathsf{gNB}}^{\mathsf{BB}}$. 
\State Repeat Steps (10-18) until the convergence of (\ref{problem1}) and (\ref{problem2}).
\State \Return $\mathbf{F}_{\mathsf{gNB}}^{\mathsf{RF}}$, $\mathbf{F}_{\mathsf{IAB}}^{\mathsf{RF}}$, $\mathbf{W}_{\mathsf{IAB}}^{\mathsf{RF}}$, $\mathbf{W}_{\mathsf{UE}}^{\mathsf{RF}}$, $\mathbf{F}_{\mathsf{gNB}}^{\mathsf{BB}} $, $\mathbf{F}_{\mathsf{IAB}}^{\mathsf{BB}} $, $\mathbf{W}_{\mathsf{IAB}}^{\mathsf{BB}} $, $\mathbf{W}_{\mathsf{UE}}^{\mathsf{BB}} $
\end{algorithmic}
\end{algorithm}

\subsection{Convergence}
In this subsection, we prove the convergence of the proposed algorithm \ref{hybrid-beamforming}. Since the digital beamforming solutions are derived in terms of the analog beamformers, the convergence of the hybrid analog/digital beamforming algorithm depends on the convergence of the analog solutions themselves. In other terms, it is sufficient to prove the convergence of the objective functions in (\ref{problem1}) and (\ref{problem2}). We show that the objective function decreases in each iteration and converges to the local optimum in a few iterations, which makes it computationally efficient. The total SI plus noise power at the IAB node, i.e., the objective function of (\ref{problem1}) is given by
\begin{equation}\label{obj1}
\begin{split}
\mathbf{R} =& \mathrm{Tr}\left( \mathbf{W}_{\mathsf{IAB}}^{\mathsf{RF}*}\mathbf{R}_{\mathsf{IAB}}\mathbf{W}_{\mathsf{IAB}}^{\mathsf{RF}} \right)\\=&
\mathrm{Tr}\left(\mathbf{W}_{\mathsf{IAB}}^{\mathsf{RF}*} \left(\rho_{\mathsf{s}}\mathbf{H}_{\mathsf{s}}\mathbf{F}_{\mathsf{IAB}}^{\mathsf{RF}} \mathbf{F}_{\mathsf{IAB}}^{\mathsf{RF}*}\mathbf{H}_{\mathsf{s}}^* + \sigma^2 \mathbf{I}_{N_\mathsf{IAB}} \right) \mathbf{W}_{\mathsf{IAB}}^{\mathsf{RF}}   \right)\\=& \underbrace{\mathrm{Tr}\left(\rho_{\mathsf{s}} \mathbf{W}_{\mathsf{IAB}}^{\mathsf{RF}*}\mathbf{H}_{\mathsf{s}}\mathbf{F}_{\mathsf{IAB}}^{\mathsf{RF}}\mathbf{F}_{\mathsf{IAB}}^{\mathsf{RF}*}\mathbf{H}_{\mathsf{s}}^*\mathbf{W}_{\mathsf{IAB}}^{\mathsf{RF}}  \right)}_{\textsf{Effective SI Power}~\left(\mathcal{J}\right)}  + \sigma^2N_{\mathsf{RF}}.
\end{split}
\end{equation}
Similarly, the SI plus noise power defined in (\ref{problem2}) is given by
\begin{equation}\label{obj2}
\begin{split}
\mathbf{S} =& \mathrm{Tr}\left( \mathbf{F}_{\mathsf{IAB}}^{\mathsf{RF}*}\mathbf{S}_{\mathsf{IAB}}\mathbf{F}_{\mathsf{IAB}}^{\mathsf{RF}} \right)\\=& \mathrm{Tr}\left(\mathbf{F}_{\mathsf{IAB}}^{\mathsf{RF}*}\left(\rho_{\mathsf{s}}\mathbf{H}_{\mathsf{s}}^*\mathbf{W}_{\mathsf{IAB}}^{\mathsf{RF}} \mathbf{W}_{\mathsf{IAB}}^{\mathsf{RF}*}\mathbf{H}_{\mathsf{s}} +\sigma^2 \mathbf{I}_{N_\mathsf{IAB}}  \right) \mathbf{F}_{\mathsf{IAB}}^{\mathsf{RF}}\right)\\=&
\mathrm{Tr}\left(\rho_{\mathsf{s}} \mathbf{F}_{\mathsf{IAB}}^{\mathsf{RF}*} \mathbf{H}_{\mathsf{s}}^* \mathbf{W}_{\mathsf{IAB}}^{\mathsf{RF}}\mathbf{W}_{\mathsf{IAB}}^{\mathsf{RF}*} \mathbf{H}_{\mathsf{s}}\mathbf{F}_{\mathsf{IAB}}^{\mathsf{RF}}\right) +\sigma^2N_{\mathsf{RF}}.
\end{split}
\end{equation}
It is noteworthy to state that the objective functions in (\ref{obj1}) and (\ref{obj2}) have the same generic form and so the solutions as well. The local optimal solutions of the objective functions (\ref{problem1}) and (\ref{problem2}) are given by Theorem \ref{theorem1} and they are sure to converge to the locally optimal solution as it is guaranteed by Algorithm \ref{hybrid-beamforming}. The effective SI power decreases in each iteration and it is lower bounded by zero. 

Fig.~\ref{plotconv} illustrates the progress of the effective SI power with respect to the number of iterations. We notice that the algorithm converges in just 10 iterations requiring 0.7 Mflops in total. In addition, we observe that the analog beamforming drops the SI power from 128 to 16 to prevent the ADC saturation (8x) while the digital beamforming drops the SI power from 16 to 10.06 (1.5x). Since the objective function (\ref{problem2}) has the same generic form as (\ref{problem1}), the results in Fig.~ \ref{plotconv} also hold for the problem (\ref{problem2}).
\begin{figure}[tb]
\vspace*{0.25in}
\centering
\setlength\fheight{5.5cm}
\setlength\fwidth{7.5cm}
%
%
\definecolor{mycolor1}{rgb}{0.00000,0.44700,0.74100}%
\definecolor{mycolor2}{rgb}{0.85000,0.32500,0.09800}%
\begin{tikzpicture}

\begin{axis}[%
width=0.951\fwidth,
height=\fheight,
at={(0\fwidth,0\fheight)},
scale only axis,
xmin=0,
xmax=20,
xlabel style={font=\color{white!15!black}},
xlabel={\textsf{Iterations}},
ymin=0,
ymax=150,
ylabel style={font=\color{white!15!black}},
ylabel={\textsf{Effective SI Power}},
axis background/.style={fill=none},
axis x line*=bottom,
axis y line*=left,
legend style={legend cell align=left, align=left, draw=none,fill=none}
]
\addplot [color=mycolor1, line width=1.5pt, mark=triangle, mark options={solid, mycolor1}]
  table[row sep=crcr]{%
0	128.675125398309\\
1	25.5978453001377\\
2	22.5332219764942\\
3	19.8092489910483\\
4	18.1219557295658\\
5	18.0439944478192\\
6	17.5700720322854\\
7	17.5404198654464\\
8	17.4442720470892\\
9	16.3629469751445\\
10	16.4119806769943\\
11	16.6634315375692\\
12	16.94008510502\\
13	16.1357927609917\\
14	16.8957045063643\\
15	16.6689649428052\\
16	16.9821834022385\\
17	16.4522536979715\\
18	17.0027093551544\\
19	15.999124315678\\
20	16.7600100802557\\
};
\addlegendentry{\textsf{Analog Cancellation} }

\addplot [color=mycolor2, line width=1.5pt, mark=o, mark options={solid, mycolor2}]
  table[row sep=crcr]{%
0	126.366886496391\\
1	15.2965150061167\\
2	15.0193699639311\\
3	12.9855063639\\
4	12.081918752509\\
5	11.6295354264083\\
6	11.4740888070268\\
7	11.3741674217781\\
8	11.4706258776203\\
9	10.9542522861326\\
10	11.142951513585\\
11	11.2621553725803\\
12	11.3884764883212\\
13	10.716446192405\\
14	11.0707605112445\\
15	11.4460562927954\\
16	11.2680039613763\\
17	11.2841191609632\\
18	11.1768313845447\\
19	10.5727338927349\\
20	11.0865054956294\\
};
\addlegendentry{\textsf{Analog and Digital Cancellation}}

\end{axis}
\end{tikzpicture}%
    \caption{Convergence of the effective SI power function implemented in analog ($\mathcal{J}$) defined in the objective function (\ref{obj1}) and hybrid analog/digital for the proposed hybrid beamforming algorithm \ref{hybrid-beamforming}. The plot is produced with $\mathsf{SNR} = 0~ \mathsf{dB}$ and SI power $\rho_{\mathsf{s}} = 15~\mathsf{dB}$.}
     \label{plotconv}
\end{figure}
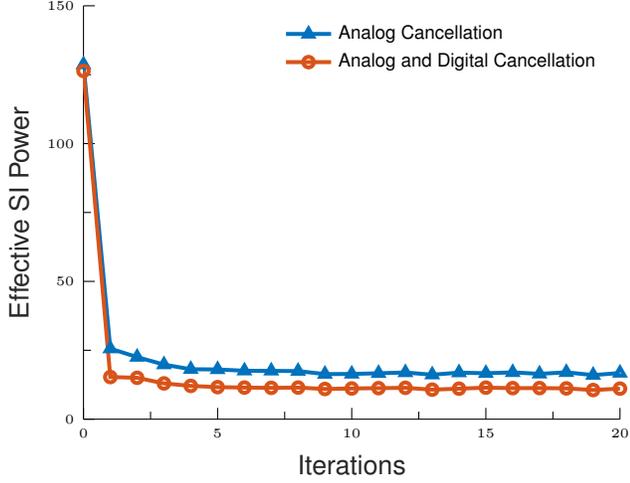


\subsection{Complexity Analysis}
Table \ref{complexity} analyzes the computational complexity of the proposed hybrid beamforming algorithm.  Multiplying matrices $\mathbf{A} \in \mathbb{C}^{m\times n}$ and $\mathbf{B} \in \mathbb{C}^{n\times p}$ requires $nmp$ flops. An inverse of an $n \times n$ matrix using Cholesky decomposition requires $\frac{n^3}{3}$ flops whereas multiplication of a matrix $\mathbf{A} \in \mathbb{C}^{m\times n}$ and its Hermitian ($\mathbf{A}\mathbf{A}^*$) requires $\frac{nm^2}{2}$ flops.

\begin{table*}[!ht]
\renewcommand{\arraystretch}{1}
\setlength{\arrayrulewidth}{.7pt}
\vspace*{0.25in}
\caption{Computational complexity of the hybrid beamforming algorithm. Parameters values are selected from Table \ref{sysparam}.}
\label{complexity}
\centering
\begin{tabular}{|c|c|r|lr|}
\hline\hline
\bfseries Operation & \bfseries Complex Multiplications for Highest-Order Terms &\bfseries  Flops &\bfseries Dominant Term & \bfseries Contribution (Total) \\
\hline
$\mathbf{W}_{\mathsf{IAB}}^{\mathsf{RF}}$& $\frac{3}{2} N_{\mathsf{IAB}}^2N_{\mathsf{RF}}^{\mathsf{IAB}} + \frac{1}{3}N_{\mathsf{IAB}}^3 + N_{\mathsf{gNB}}N_{\mathsf{IAB}}N_{\mathsf{RF}}^{\mathsf{gNB}}+\frac{1}{3}\left(N_{\mathsf{RF}}^{\mathsf{gNB}}\right)^3 + N_{\mathsf{IAB}}^2N_\mathsf{RF}^{\mathsf{gNB}}$ & 21165 & $\frac{1}{3}N_{\mathsf{IAB}}^3$ & $51.61\%~(16.06\%)$ \\
\hline
$\mathbf{F}_{\mathsf{IAB}}^{\mathsf{RF}}$& $\frac{3}{2} N_{\mathsf{IAB}}^2N_{\mathsf{RF}}^{\mathsf{IAB}} + \frac{1}{3}N_{\mathsf{IAB}}^3 + N_{\mathsf{UE}}N_{\mathsf{IAB}}N_{\mathsf{RF}}^{\mathsf{UE}}+\frac{1}{3}\left(N_{\mathsf{RF}}^{\mathsf{UE}}\right)^3 + N_{\mathsf{IAB}}^2N_\mathsf{RF}^{\mathsf{UE}}$& 19373 & $\frac{1}{3}N_{\mathsf{IAB}}^3$ & $56.38\%~(16.06\%)$\\
\hline
$\mathbf{F}_{\mathsf{gNB}}^{\mathsf{RF}}$& $\frac{3}{2}N_{\mathsf{gNB}}^2N_{\mathsf{RF}}^{\mathsf{gNB}} + \frac{1}{3}N_{\mathsf{gNB}}^3$& 13995 & $\frac{1}{3}N_{\mathsf{gNB}}^3$ & $78.05\%~(16.06\%)$\\
\hline
$\mathbf{W}_{\mathsf{IAB}}^{\mathsf{BB}}$ & $9\left(N_{\mathsf{RF}}^{\mathsf{IAB}}\right)^2N_{\mathsf{IAB}} + 9N_{\mathsf{s}}^2N_{\mathsf{IAB}} + N_{\mathsf{IAB}}^2 N_{\mathsf{s}} + N_{\mathsf{s}}^3$  & 4360 & $N_{\mathsf{IAB}}^2 N_{\mathsf{s}}$ & $46.97\%~(3.03\%)$ \\
\hline
$\mathbf{F}_{\mathsf{IAB}}^{\mathsf{BB}}$ & $9\left(N_{\mathsf{RF}}^{\mathsf{IAB}}\right)^2N_{\mathsf{IAB}} + 9N_{\mathsf{s}}^2N_{\mathsf{IAB}} + N_{\mathsf{IAB}}^2 N_{\mathsf{s}} + N_{\mathsf{s}}^3$  & 4360 &$N_{\mathsf{IAB}}^2 N_{\mathsf{s}}$ & $46.97\%~(3.01\%)$ \\
\hline
$\mathbf{F}_{\mathsf{gNB}}^{\mathsf{BB}}$ & $9\left(N_{\mathsf{RF}}^{\mathsf{gNB}}\right)^2N_{\mathsf{gNB}} + 9N_{\mathsf{s}}^2N_{\mathsf{gNB}} + N_{\mathsf{gNB}}^2 N_{\mathsf{s}} + N_{\mathsf{s}}^3$  & 4360 & $N_{\mathsf{gNB}}^2 N_{\mathsf{s}}$ & $46.97\%~(3.01\%)$\\
\hline
$\mathbf{W}_{\mathsf{UE}}^{\mathsf{BB}}$ & $9\left(N_{\mathsf{RF}}^{\mathsf{UE}}\right)^2N_{\mathsf{UE}} + 9N_{\mathsf{s}}^2N_{\mathsf{UE}} + N_{\mathsf{UE}}^2 N_{\mathsf{s}} + N_{\mathsf{s}}^3$  & 328 & $9\left(N_{\mathsf{RF}}^{\mathsf{UE}}\right)^2N_{\mathsf{UE}}$ & $43.90\%~(0.21\%)$ \\
\hline
$\mathbf{W}_{\mathsf{UE}}^{\mathsf{RF}}$& $\frac{3}{2}N_{\mathsf{UE}}^2N_{\mathsf{RF}}^{\mathsf{UE}} + \frac{1}{3}N_{\mathsf{UE}}^3$& 70 & $\frac{3}{2}N_{\mathsf{UE}}^2N_{\mathsf{RF}}^{\mathsf{UE}}$ & $69.23\%~(0.07\%)$\\
\hline\hline
\end{tabular}
\end{table*}



\section{Numerical Analysis}
Table \ref{sysparam} gives the parameter values used in the simulations.  For each case, 1000 channels realizations were generated to perform the Monte Carlo simulation in MATLAB.
\begin{table}[t]
\renewcommand{\arraystretch}{1}
\setlength{\arrayrulewidth}{.7pt}
\caption{System parameters.}
\label{sysparam}
\centering
\begin{tabular}{|l|c|}
\hline\hline
\bfseries Parameter & \bfseries Value\\
\hline
Carrier frequency& 28 GHz \\
\hline
Bandwidth & 850 MHz\\
\hline
Number of gNB/IAB Antennas ($N_{\mathsf{gNB}}/N_{\mathsf{IAB}}$) & 32\\
\hline
Number of UE Antennas ($N_{\mathsf{UE}}$) & 4 \\
\hline
Number of Clusters ($C$) & 6\\
\hline
Number of Rays per Cluster ($R_c$) & 8\\
\hline
AoA/AoD Angular Spread& 20$^{\circ}$ \\
\hline
Transceivers Gap ($d$)& 2$\lambda$\\
\hline
Transceivers Incline ($\omega$) & $\frac{\pi}{6}$ \\
\hline
Rician Factor ($\kappa$) & 5 dB\\
\hline
SI Power ($\rho_{\mathsf{s}}$) & 15 dB \\ 
\hline
Number of Spatial Streams ($N_{\mathsf{s}}$)& 2\\ 
\hline
Number of RF Chains ($N_{\mathsf{RF}}$)  & 2\\
\hline\hline
\end{tabular}
\end{table}

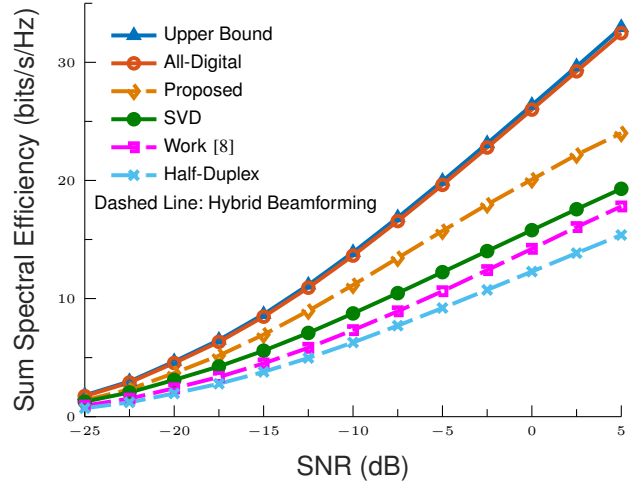
\begin{figure}[tb]
\centering
\vspace*{0.25in}
\setlength\fheight{5.5cm}
\setlength\fwidth{7.5cm}
%
%
\definecolor{mycolor1}{rgb}{0.00000,0.44700,0.74100}%
\definecolor{mycolor2}{rgb}{0.85000,0.32500,0.09800}%
\definecolor{mycolor3}{rgb}{0.92900,0.69400,0.12500}%
\definecolor{mycolor4}{rgb}{0.49400,0.18400,0.55600}%
\definecolor{mycolor5}{rgb}{0.46600,0.67400,0.18800}%
\definecolor{mycolor6}{rgb}{0.30100,0.74500,0.93300}%

\definecolor{mycolor7}{rgb}{1.00000,0.00000,1.00000}%
\definecolor{mycolor8}{rgb}{0.00000,0.49804,0.00000}%
\definecolor{mycolor9}{rgb}{0.00000,0.44706,0.74118}%
\definecolor{mycolor10}{rgb}{0.87059,0.49020,0.00000}%
\begin{tikzpicture}

\begin{axis}[%
width=0.951\fwidth,
height=\fheight,
at={(0\fwidth,0\fheight)},
scale only axis,
xmin=-25,
xmax=5,
xlabel style={font=\color{white!15!black}},
xlabel={\textsf{SNR (dB)}},
ymin=0,
ymax=35,
ylabel style={font=\color{white!15!black}},
ylabel={\textsf{Sum Spectral Efficiency (bits/s/Hz)}},
axis background/.style={fill=white},
axis x line*=bottom,
axis y line*=left,
legend style={at={(0.03,0.97)},anchor=north west,legend cell align=left, align=left, draw=none}
]

\node[right, align=left, rotate=0]
at (axis cs:-25,18) {\scriptsize \textsf{Dashed Line: Hybrid Beamforming} };

\addplot [color=mycolor1,mark=triangle, line width=1.5pt]
  table[row sep=crcr]{%
-25	1.80839514961157\\
-22.5	2.98148134909461\\
-20	4.67025099614668\\
-17.5	6.48896559219936\\
-15	8.66148342617481\\
-12.5	11.1521687136973\\
-10	13.9048716017374\\
-7.5	16.8582377153452\\
-5	19.9447443988196\\
-2.5	23.1294496396274\\
0	26.3739419546404\\
2.5	29.6441846105043\\
5	32.9496839071419\\
};
\addlegendentry{\textsf{Upper Bound}}

\addplot [color=mycolor2,mark = o, line width=1.5pt]
  table[row sep=crcr]{%
-25	1.73998618578106\\
-22.5	2.88862011153832\\
-20	4.54641771446955\\
-17.5	6.33079196331291\\
-15	8.46669524994564\\
-12.5	10.9192215419585\\
-10	13.636356140291\\
-7.5	16.5612333962715\\
-5	19.6267187769265\\
-2.5	22.7881938508954\\
0	25.9981341612962\\
2.5	29.2381277657277\\
5	32.4678434722871\\
};
\addlegendentry{\textsf{All-Digital}}

\addplot [color=mycolor10,mark=diamond,mark size = 2.5,dash pattern={on 10pt off 3pt on 0pt off 0pt}, line width=1.5pt]
  table[row sep=crcr]{%
-25	1.38676086201304\\
-22.5	2.31585385244829\\
-20	3.69153405389599\\
-17.5	5.17242248825439\\
-15	6.94665357447858\\
-12.5	8.95968170082838\\
-10	11.1464621124065\\
-7.5	13.4328180543866\\
-5	15.7252870890375\\
-2.5	17.9617152318461\\
0	20.0795312224624\\
2.5	22.1739552107952\\
5	24.0031199674035\\
};
\addlegendentry{\textsf{Proposed}}

\addplot [color=mycolor8,mark =*, line width=1.5pt]
  table[row sep=crcr]{%
-25	1.27196370446014\\
-22.5	2.05066763170641\\
-20	3.11202639068595\\
-17.5	4.25088841129759\\
-15	5.58316406447165\\
-12.5	7.08984599820367\\
-10	8.74255295500149\\
-7.5	10.4608219978856\\
-5	12.2353892050461\\
-2.5	14.0190743386121\\
0	15.7912001102845\\
2.5	17.5643304982342\\
5	19.2911859031953\\
};
\addlegendentry{\textsf{SVD}}

\addplot [color=mycolor7,mark =square, mark size= 1.5,dash pattern={on 10pt off 3pt on 0pt off 0pt}, line width=1.5pt]
  table[row sep=crcr]{%
-25	0.935370276796865\\
-22.5	1.52913026606292\\
-20	2.40438892884529\\
-17.5	3.35976661514808\\
-15	4.50744038560152\\
-12.5	5.82966944406655\\
-10	7.32581053263052\\
-7.5	8.92130740165908\\
-5	10.6197068837029\\
-2.5	12.4086215026508\\
0	14.2100258839061\\
2.5	16.0524619285101\\
5	17.7998434118524\\
};
\addlegendentry{\textsf{Work} \cite{ian}}

\addplot [color=mycolor6,mark =x , mark size = 2.5,dash pattern={on 10pt off 3pt on 0pt off 0pt}, line width=1.5pt]
  table[row sep=crcr]{%
-25	0.717214697534494\\
-22.5	1.2109212862488\\
-20	1.96408873601977\\
-17.5	2.78711561763153\\
-15	3.79311680119515\\
-12.5	4.96689802088497\\
-10	6.2812794549039\\
-7.5	7.70625910353802\\
-5	9.19959397968485\\
-2.5	10.733233610917\\
0	12.2836794061653\\
2.5	13.844637265403\\
5	15.3814916496101\\
};
\addlegendentry{\textsf{Half-Duplex}}

\end{axis}
\end{tikzpicture}%
    \caption{Sum spectral efficiency results: Performance comparison between the proposed algorithm with the related works as well as the benchmarking tools.}
     \label{plot1}
\end{figure}

Among the three FD hybrid beamforming solutions in Fig.~\ref{plot1}, SVD and \cite{ian} are very sensitive to SI because the relative analog beamformers ignore SI cancellation, which leads to ADC saturation and hence more performance loss. Our approach, however, introduces analog beamformer design to reduce a large amount of SI power as shown in Fig.~\ref{plotconv}. The performance of the proposed system is improved by the optimal digital beamforming solution which further suppresses residual SI. At $\mathsf{SNR} = 5~\mathsf{dB}$, we notice the proposed FD system achieves a gain of around 4.71, and 6.2 bits/s/Hz with respect to SVD, work \cite{ian}, respectively. In addition, our proposed FD beamforming algorithm outperforms the HD mode, which is a goal of this work, and achieves a gain of 8.62 bits/s/Hz at $\mathsf{SNR} = 5~\mathsf{dB}$.

\section{Conclusion}
In this paper, we proposed a low complexity hybrid analog/digital beamforming design for a full duplex integrated access and backhaul system. The proposed algorithm designs the hybrid precoders for the gNB Donor and IAB node, and the hybrid combiners for the IAB node and user equipment.  In simulation, the algorithm converges in five iterations while reducing a large amount of SI in the analog domain to avoid ADC saturation. In addition, the hybrid beamforming results are further improved by the implementation of the optimal digital beamformers which wipe out the residual SI. Simulations show that the proposed FD beamforming design outperforms the related works in terms of spectral efficiency as well as it beats the half duplex mode which demonstrates the feasibility of the proposed design for practical consideration.

\appendices
\section{Proof of Theorem \ref{theorem1}}\label{appendixtheorem1}
The objective is to minimize the equality in (\ref{problem1}), while preserving the signal dimensions. We begin by expressing the Lagrangian function given by
\begin{equation}
\mathcal{L}\left(\mathbf{W}_{\mathsf{IAB}}^{\mathsf{RF}},x\right) = \mathbf{W}_{\mathsf{IAB}}^{\mathsf{RF}*}\mathbf{R}_{\mathsf{IAB}}\mathbf{W}_{\mathsf{IAB}}^{\mathsf{RF}}  + x\left(\mathbf{W}_{\mathsf{IAB}}^{\mathsf{RF}*}\mathbf{H}_{\mathsf{b}}\mathbf{F}_{\mathsf{gNB}}^{\mathsf{RF}} - \mathbf{I}_{N_\mathsf{RF}}\right)    
\end{equation}
The Lagrangian conditions for this problem are
\begin{equation}\label{lag}
 \nabla_{\mathbf{W}_{\mathsf{IAB}}^{\mathsf{RF}*}}\mathcal{L}=0     
\end{equation}
\begin{equation}\label{lag1}
x^\star \left(\mathbf{W}_{\mathsf{IAB}}^{\mathsf{RF}*}\mathbf{H}_{\mathsf{b}}\mathbf{F}_{\mathsf{gNB}}^{\mathsf{RF}} - \mathbf{I}_{N_\mathsf{RF}}\right) = 0   
\end{equation}
Eq. (\ref{lag}) can be reformulated as 
\begin{equation}\label{eqlag}
\begin{split}
 &\nabla_{\mathbf{W}_{\mathsf{IAB}}^{\mathsf{RF}*}}\mathrm{Tr}\left( \mathbf{W}_{\mathsf{IAB}}^{\mathsf{RF}*}\mathbf{R}_{\mathsf{IAB}}\mathbf{W}_{\mathsf{IAB}}^{\mathsf{RF}}\right) \\&+ x^\star \nabla_{\mathbf{W}_{\mathsf{IAB}}^{\mathsf{RF}*}}\left(\mathbf{W}_{\mathsf{IAB}}^{\mathsf{RF}*}\mathbf{H}_{\mathsf{b}}\mathbf{F}_{\mathsf{gNB}}^{\mathsf{RF}} - \mathbf{I}_{N_\mathsf{RF}}\right) = 0
 \end{split}
\end{equation}
Where $\nabla$ is the gradient operator and $x^\star$ is the Lagrangian multiplier. Differentiating (\ref{eqlag}) with respect to $\mathbf{W}_{\mathsf{IAB}}^{\mathsf{RF}*}$, we get
\begin{equation}
 \mathbf{R}_{\mathsf{IAB}}\mathbf{W}_{\mathsf{IAB}}^{\mathsf{RF}} + x \mathbf{H}_{\mathsf{b}}\mathbf{F}_{\mathsf{gNB}}^{\mathsf{RF}} = 0  
\end{equation}
\begin{equation}
\mathbf{W}_{\mathsf{IAB}}^{\mathsf{RF}} = -  \mathbf{R}_{\mathsf{IAB}}^{-1}\mathbf{H}_{\mathsf{b}}\mathbf{F}_{\mathsf{gNB}}^{\mathsf{RF}}x   
\end{equation}
Then substituting the expression of ($\mathbf{W}_{\mathsf{IAB}}^{\mathsf{RF}}$) in (\ref{lag1}), we obtain
\begin{equation}
\left(- \mathbf{R}_{\mathsf{IAB}}^{-1}\mathbf{H}_{\mathsf{b}}\mathbf{F}_{\mathsf{gNB}}^{\mathsf{RF}}x  \right)^*    \mathbf{H}_{\mathsf{b}}\mathbf{F}_{\mathsf{gNB}}^{\mathsf{RF}} = \alpha \mathbf{I}_{N_{\mathsf{RF}}}
\end{equation}
\begin{equation}
x = -\alpha \left( \mathbf{F}_{\mathsf{gNB}}^{\mathsf{RF}*}\mathbf{H}_{\mathsf{b}}^*\mathbf{R}_{\mathsf{IAB}}^{-1} \mathbf{H}_{\mathsf{b}}\mathbf{F}_{\mathsf{gNB}}^{\mathsf{RF}}  \right)^{-1}    
\end{equation}
Thereby
\begin{equation}
\mathbf{W}_{\mathsf{IAB}}^{\mathsf{RF}} = \alpha \mathbf{R}_{\mathsf{IAB}}^{-1}\mathbf{H}_{\mathsf{b}}\mathbf{F}_{\mathsf{gNB}}^{\mathsf{RF}}\left( \mathbf{F}_{\mathsf{gNB}}^{\mathsf{RF*}}\mathbf{H}_{\mathsf{b}}^* \mathbf{R}_{\mathsf{IAB}}^{-1} \mathbf{H}_{\mathsf{b}} \mathbf{F}_{\mathsf{gNB}}^{\mathsf{RF}}   \right)^{-1}        
\end{equation}
The proof of the analog precoder at the IAB node ($\mathbf{F}_{\mathsf{IAB}}^{\mathsf{RF}}$) follows the same derivation steps as ($\mathbf{W}_{\mathsf{IAB}}^{\mathsf{RF}}$).
\section{Proof of Theorem \ref{theorem2}}\label{appendixtheorem2}
For $\mathbf{X}_{\mathsf{RF}}$ given
Consider the SVD $\mathbf{X}_{\mathsf{RF}} = \mathbf{U}_{\mathsf{RF}}\mathbf{S}_{\mathsf{RF}}\mathbf{V}_{\mathsf{RF}}^*$, and let $\mathbf{Q} = \mathbf{S}_{\mathsf{RF}}\mathbf{V}_{\mathsf{RF}}^*\mathbf{X}_{\mathsf{BB}} \in \mathbb{C}^{L\times N}$. Then $\mathbf{X}_{\mathsf{RF}}\mathbf{X}_{\mathsf{BB}}=\mathbf{U}_{\mathsf{RF}}\mathbf{Q}$ so that $\mathbf{X}_{\mathsf{BB}}^*\mathbf{X}_{\mathsf{RF}}^*\mathbf{X}_{\mathsf{RF}}\mathbf{X}_{\mathsf{BB}}=\mathbf{Q}^*\mathbf{U}_{\mathsf{RF}}^*\mathbf{U}_{\mathsf{RF}}\mathbf{Q}=\mathbf{Q}^*\mathbf{Q}$. The generic form of the spectral efficiency in (\ref{ratebackhaul}) and (\ref{rateaccess}) is expressed in terms of $\mathbf{Q}$ as 
\begin{equation}\label{opt2}
\begin{split}
 \max\limits_{\mathbf{Q}^*\mathbf{Q}=\mathbf{I}_N} \log\det\left( \mathbf{I}_N + \rho \mathbf{Q}^*\mathbf{U}_{\mathsf{RF}}^*\mathbf{A}\mathbf{A}^*\mathbf{U}_{\mathsf{RF}}\mathbf{Q} \right)
\end{split}
\end{equation}
Solution $\mathbf{Q}_{\star}$ is given by the $N$ dominant left singular vectors of $\mathbf{U}_{\mathsf{RF}}\mathbf{A}$. By changing variables, we solve $\mathbf{X}_{\mathsf{BB}}=\mathbf{V}_{\mathsf{RF}}\mathbf{S}_{\mathsf{RF}}^{-1}\mathbf{Q}_{\star}$. For $\mathbf{Q}=\mathbf{Q}_{\star}$, the objective function becomes
\begin{equation}\label{opt3}
\begin{split}
\log\det\left( \mathbf{I}_N + \rho \mathbf{Q}^*\mathbf{U}_{\mathsf{RF}}^*\mathbf{A}\mathbf{A}^*\mathbf{U}_{\mathsf{RF}}\mathbf{Q} \right) \leq \log\det\left( \mathbf{I}_N + \rho \mathbf{A}^*\mathbf{A}\right)
\end{split}
\end{equation}
with the bound in (\ref{opt3}) applying to any semi-unitary $\mathbf{U}_{\mathsf{RF}}$. This bound holds with equality if the columns of $\mathbf{U}_{\mathsf{RF}}$ are taken as the $L$ dominant left singular vectors of $\mathbf{A}$.

\bibliographystyle{IEEEtran}
\bibliography{main}

\end{document}